\documentclass{article}
\usepackage{tabularx,array}
\usepackage{amssymb}
\usepackage{amsmath, amsthm, amsfonts,bm}
\usepackage{graphicx}
\usepackage{color}
\usepackage{lineno}
\usepackage{environ}
\usepackage{thm-restate}
\usepackage{enumerate}

\usepackage{xcolor}

\newcommand{\crs}{\overline{\operatorname{cr}}}

\newcounter{quote}

\NewEnviron{myquote}{
	\vspace{0.2cm}
	\refstepcounter{quote}%
	\parbox{13cm\textwidth0cm}%
	{\small\textit{\BODY}}%
	\vspace{0.2cm}
	}
\newcommand{\MyQuote}[1]{\vspace{1cm}
     \parbox{10cm}{\em #1}\hspace*{2cm}($\ast$)\\[1cm]}

     \newcommand{\NMyQuote}[1]{\vspace{1cm}
     \parbox{10cm}{\em #1}\\[1cm]}

\newtheorem{theorem}{Theorem}[section]

\newtheorem{lemma}[theorem]{Lemma}

\theoremstyle{definition}

\theoremstyle{remark}

\begin{document}

\title{A Note on the $k$-colored Crossing Ratio of Dense Geometric Graphs}

\author{Ruy~Fabila-Monroy\thanks{Departamento de Matem\'aticas, CINVESTAV, 
{\tt ruyfabila@math.cinvestav.edu.mx}, Partially supported by CONACYT FORDECYT-PRONACES/39570/2020}}

\maketitle              %
\begin{abstract}
A \emph{geometric graph} is a graph whose vertex set is a set of points in general position in the plane,
and its edges are straight line segments joining these points. We show that for every integer $k \ge 2$,
there exists a constant $c>0$ such that the following holds.
The edges of every dense geometric graph, with sufficiently many vertices, can be colored with $k$ colors, such that
the number of pairs of edges  of the same color that cross is at most $(1/k-c)$ times the total
number of pairs of edges  that cross. The case when $k=2$ and
$G$ is a complete geometric graph, was proved by Aichholzer et al.[\emph{GD} 2019].
\end{abstract}

\section{Introduction}

A \emph{geometric graph}, $G=(V,E)$, is a graph whose vertex set is a set of points in general position\footnote{no three of them collinear} in the plane,
and its edges are straight line segments joining these points. Let $\crs(G)$ be the number of pairs
of edges of $G$ that cross. Let $\chi$ be an edge-coloring of $G$. If $\chi$ uses $k$ colors,
we say that it is a \emph{$k$-coloring}. Let $\crs(G,\chi)$ be the number of pairs
of edges of $G$ that cross and that are of the \emph{same color} in $\chi$. Let $k$ be a positive integer, and let $\chi$ be a 
$k$-coloring of the edges of $G$, in which each edge of $G$ is assigned one of $k$ colors
independently and uniformly at random. The probability that a given pair of crossing edges 
receive the same color is equal to $1/k$; thus,
$E[\crs(G,\chi)]=\frac{1}{k}\crs (G).$
Therefore, there exists a choice of $\chi$ for which \[\frac{\crs(G,\chi)}{\crs (G)} \le \frac{1}{k}.\]

Aichholzer et. al. \cite{2colored} showed that for the case when $G$ is a complete geometric graph,
there exists a constant $c >0$ (independent of $G$) and a $2$-coloring, $\chi$, of the edges
of $G$ such that 
\[\frac{\crs(G,\chi)}{\crs (G)} \le \frac{1}{2}-c.\]
A \emph{dense} graph on $n$ vertices is a graph with at least $d \binom{n}{2}$ edges, for some
positive constant $d$; $d$ is called the \emph{density} of $G$, and we denote it with $d(G)$.
In this paper we generalize the result of~\cite{2colored} to
dense geometric graphs and for $k$-colorings with $k\ge 2$. Specifically, we show the following.
\begin{theorem}\label{thm:main}
 Let $G$ be a dense geometric graph, with sufficiently many vertices, and with density $d>0$. For every integer $k \ge 2$, there exists
 a positive constant $c=c(d,k)$ (depending only on $d$ and $k$) and a $k$-coloring, $\chi$, of the edges
 of $G$ such that 
 \[\frac{\crs(G,\chi)}{\crs (G)} \le \frac{1}{k}-c.\]
\end{theorem}
\qed

To prove Theorem~\ref{thm:main} we use some results from extremal graph theory and combinatorial geometry.
For completeness, we present them along the way. We follow the expositions
of Diestel's~\cite{diestel} and Matou{\v{s}}ek's~\cite{mat} books.

\section{Preliminaries}

\subsection*{Proof of Theorem~\ref{thm:main} for the case $k=2$}

For $X,Y \subset V(G)$, let $E(X,Y)$ be the set of edges of $G$ that have an endpoint
in $X$ and an endpoint in $Y$; we call them \emph{$X-Y$ edges}. A pair of 
edges in $G$ is called \emph{monochromatic} if they are of the same color; otherwise,
it is called \emph{heterochromatic}.

Before proceeding,
it is convenient to give a high level overview of the steps of the proof in~\cite{2colored},
for the case when $G$ is a complete geometric graph on $n$ vertices, and $k=2$.
They are as follows. 

\begin{itemize}
\item Show that there exists subsets $Y_1,Y_2,Z_1,Z_2$ of vertices
of $G$ such that:
\begin{itemize}
 \item  every $Y_1-Z_1$ edge crosses every $Y_2-Z_2$ edge; and
 \item each $Y_i$ and each $Z_i$ has $c'n$ points for some positive constant $c'$. 
\end{itemize}

\item Color all the $Y_1-Z_1$ edges with ``red''
and all  $Y_2-Z_2$ edges with ``blue''. Let $E':=E(Y_1,Z_1) \cup E(Y_2,Z_2)$.
The number of monochromatic pairs of crossing edges in $E'$ is equal to
\[2\cdot \left(\binom{c'n}{2} \cdot \binom{c'n}{2} \right ) \approx \frac{(c'n)^4}{2}.\]
The number of heterochromatic pairs of crossing edges in $E'$ is equal to
\[(c'n)^4.\] Therefore, at most $1/3$ of the crossings between the edges in $E'$
are monochromatic.

\item Finally, color the remaining  set of edges, $E'':=E \setminus E'$, uniformly and
independently at  random with ``red'' or ``blue''.
Let $C_1$ be the set of pairs of edges in $E''$ that cross. Let $C_2$
be the set of pairs of edges, one in $E'$ and the other in $E''$, that cross.
The probability that a given pair in $C_1 \cup C_2$ is monochromatic is equal to $1/2$.
By linearity of expectation, we have that the expected number of monochromatic pairs
in $C_1 \cup C_2$ is equal to $\frac{1}{2}|C_1|+\frac{1}{2}|C_2|$.
Therefore, there exists a $2$-coloring with at most this number of crossings.
Fix this $2$-coloring, and let $\chi$ be the resulting coloring of $E=E'\cup E''$.
 The $2$-coloring so constructed satisfies that \[\frac{\crs(G,\chi)}{\crs(G)} \le \frac{1}{2}-c,\]
for some constant $c > 0$ depending only on $c'$.
\end{itemize}

There are two possible problems when trying to generalize this approach to the case when
$G$ is not complete:
\begin{itemize}
 \item[(1)] The number of edges in $E'$ might be significantly smaller than the number of edges in $G$.

\item[(2)] Many of the edges in $E(Y_1,Z_1)$ might cross each other, and many of the edges 
in $E(Y_2,Z_2)$ might cross each other,
compared to the number of crossings between an edge in $E(Y_1,Z_1)$ and an edge in $E(Y_2,Z_2)$.
\end{itemize}
(1) implies that even if we manage to color the edges in $E'$ in a good way this might have
little impact on ${\crs(G, \chi)}/{\crs(G)}$ in the end. The problem with (2) is as follows.
Suppose that $|E(Y_1,Z_1)|=|E(Y_2,Z_2)|=:m$; and
that every $Y_1-Z_1$ edge crosses every other $Y_1-Z_1$ edge, and that the same
holds for the $Y_2-Z_2$ edges. If we color $E(Y_1,Z_1) \cup E(Y_2,Z_2)$ as above, then the number of 
monochromatic pair of crossing edges is equal to 
\[\binom{m}{2}+\binom{m}{2} \approx m^2.\]
While, the number of heterochromatic pairs of crossing edges is equal to  \[m^2.\] 
Thus, the number of monochromatic crossings pairs of edges in $E'$ 
is asymptotically $1/2$ of the total number of crossing pairs of edges.

In addition, to generalize this approach to the case when $k >2$, we also need to show that:
\begin{itemize}
 \item[(3)] there exists subsets $Y_1,\dots,Y_k,Z_1,\dots,Z_k$ of vertices
of $G$ such that: 
\begin{itemize}
    \item every $Y_i-Z_i$ edge crosses every $Y_j-Z_j$ edge, for every pair $i \neq j$; and
    \item each $Y_i$ and each $Z_i$ has $c'n$ points for some positive constant $c'$. 
\end{itemize}
\end{itemize}

Having addressed these issues, Theorem~\ref{thm:main} can be derived from the
following lemma.
\begin{lemma}\label{lem:main}
 Let $k \ge 2$ be an integer, and $G=(V,E)$ be a geometric graph on $n$ vertices, with $n$ sufficiently large, and density $d >0$.
 Suppose that there exist positive constants $c_1,c_2$ and $c_3<c_2^2/2$, depending only on $d$ and $k$, such that
 the following hold.
 \begin{itemize}
    \item[$a)$] There exists subsets  $Y_1,\dots,Y_k,Z_1,\dots,Z_k$ of vertices  of $G$, each with $c_1n$ points;
    \item[$b)$] $|E(Y_i,Z_i)| \ge c_2 n^2$ for every $1 \le i \le k$;
    \item[$c)$] every $Y_i-Z_i$ edge crosses every $Y_j-Z_j$ edge, for every pair $i \neq j$;
    \item[$d)$] the number of pairs of $Y_i-Z_i$ edges that cross is at most $(c_2^2/2-c_3)n^4$, for every $1 \le i \le k$.
 \end{itemize}
 Then there  exists
 a positive constant $c=c(d,k)$ (depending only on $d$ and $k$) and a $k$-coloring, $\chi$, of the edges
 of $G$ such that 
 \[\frac{\crs(G,\chi)}{\crs (G)} \le \frac{1}{k}-c.\]
\end{lemma}
\begin{proof}
 Let 
 \[E'=\bigcup_{i=1}^k E(Y_i,Z_i),\] 
 and let $G'$ be the geometric graph with vertex set equal to $V$ and edge set equal to $E'$.
 Let $\chi'$ be the edge coloring of $G'$, in which all the $Y_i-Z_i$ edges receive color $i$.
 For every $1 \le i \le k$, let $s_i$ be the number of pairs of $Y_i-Z_i$ edges that cross.
 Thus, $s_i \le (c_2^2/2-c_3)n^4.$ 
 The number of heterochromatic crossings pairs of edges  in $G'$ is at least \[\frac{k(k-1)}{2}c_2^2n^4;\] 
 and the number of monochromatic crossings pairs of edges in $G'$ 
 is equal to \[\sum_{i=1}^k s_i \le k\left (\frac{c_2^2}{2}-c_3 \right )n^4.\]
 
 We have that 
 \[\frac{\crs(G',\chi')}{\crs (G')} \le \frac{\sum_{i=1}^k s_i}{\frac{k(k-1)}{2}c_2^2n^4+\sum_{i=1}^k s_i}.\]
 This is maximized when $\sum_{i=1}^k s_i$ is maximized.
 Therefore, 
 \begin{align*} 
 \frac{\crs(G',\chi')}{\crs (G')}  & \le \frac{ k\left (\frac{c_2^2}{2}-c_3 \right )n^4}{\frac{k(k-1)}{2}c_2^2n^4+k\left (\frac{c_2^2}{2}-c_3 \right )n^4} \\
 & = \frac{\frac{c_2^2}{2}-c_3}{k\left (\frac{c_2^2}{2}-\frac{1}{k}c_3 \right)} \\
 & =\frac{\frac{c_2^2}{2}-\frac{1}{k}c_3}{k\left (\frac{c_2^2}{2}-\frac{1}{k}c_3 \right)}-\frac{c_3-\frac{1}{k}c_3}{k\left (\frac{c_2^2}{2}-\frac{1}{k}c_3 \right)}\\
 & =\frac{1}{k}-\frac{c_3-\frac{1}{k}c_3}{k\frac{c_2^2}{2}-c_3}\\
 & =\frac{1}{k}-c',
 \end{align*}
 with $c':=\left (c_3-\frac{1}{k}c_3 \right ) /\left (k\frac{c_2^2}{2}-c_3 \right )$. Since $c_3-\frac{1}{k}c_3 >0$ and $k\frac{c_2^2}{2}-c_3>0$, we have
 that $c'>0$. 
 
 Let $E'':=E\setminus E'$. Let $C_1$ be the set of pairs of edges in $E''$ that cross.
 Let $C_2$ be the set of pairs of edges, consisting of an edge in $E''$ and an edge in $E'$, that cross.
 By the previous probabilistic argument and linearity of expectation, there  exists a $k$-coloring, $\chi''$, of $E''$ such
 that the number of monochromatic pairs in $C_1 \cup C_2$ is at most $|C_1|/k +|C_2|/k$.
 Let $\chi:=\chi'\cup \chi''$.
 We have that
\begin{eqnarray}
 \frac{\crs(G,\chi)}{\crs(G)} & \le & \frac{\crs(G',\chi')+C_1/k+C_2/k}{\crs(G')+C_1+C_2} \nonumber \\
                                & \le & \frac{(1/k-c')\crs(G')+C_1/k+C_2/k}{\crs(G')+C_1+C_2} \nonumber \\   
                                & = & \frac{1}{k}-c' \cdot \frac{\crs(G')}{\crs(G')+C_1+C_2} \nonumber \\
                                & \le & \frac{1}{k}-c, \nonumber
\end{eqnarray}
with $c:=c' \cdot \frac{\crs(G')}{\crs(G')+C_1+C_2}$.
By the Crossing number theorem(see ~\cite{mat}), every
dense geometric graph on $n$ vertices has $\Theta(n^4)$ crossings. Thus,
$\crs(G'), C_1, C_2 $ are $\Theta(n^4)$, and $c >0$.
\end{proof}

In what follows, let $G:=(V,E)$ be a dense geometric graph on $n$ vertices and density equal to $d$, and  let $k\ge 2$.
We now give the necessary background needed to show that conditions $a),b),c)$ and $d)$ of Lemma~\ref{lem:main}
hold for $G$.

\subsection*{The Same Type Lemma}

Let $S$ be a set of $n$ points in general position  in the plane.
To every triple $(p,q,r)$ of points of $S$ assign a  $``-"$ if $r$ is to the left
of the directed line from $p$ to $q$, and assign a $``+"$ if $r$ lies to the right 
of the directed line from $p$ to $q$. This assignment is called the \emph{order type} of $S$. 
Order types were introduced  by Goodman and Pollack~\cite{m_sorting}. They serve 
as a combinatorial abstraction of the convex hull containment relationships of points sets. 
Let $P$ and $Q$ be two sets of $n$ points in general position in the plane. Let $f$
be a bijection from $P$ to $Q$. We say that $f$ \emph{preserves the order type}
if every triple $(p,q,r)$ of points in $P$ has the same sign as $(f(p),f(q),f(r))$.
If such an $f$ exists we say that $P$ and $Q$ have the same order type.
 In this case, two edges with endpoints in
 $Q$ cross if and only if the corresponding edges in $P$ cross.

Let $(X_1,..., X_t)$ be a tuple of finite disjoint sets of points in the plane,
such that $\bigcup_{i=1}^t X_i$ is in general position.
A \emph{transversal} of $(X_1,..., X_t)$ is a tuple of points
$(x_1,\dots,x_t)$ such that $x_i \in X_i$, for all $i$.
We say that $(X_1,..., X_t)$ has \emph{same-type transversals} if  the following holds.
For every two of its transversals $(x_1,\dots, x_t)$
and $(x_1',\dots, x_t')$, the mapping $x_i \mapsto x_i'$ preserves
the order type between $\{x_1,\dots, x_t\}$ and $\{x_1',\dots, x_t'\}$. 
B\'ar\'any and Valtr~\cite{positive_es} proved the following.
\begin{theorem}[Same-type lemma]\label{thm:same_type}
For every positive integer $t$ there exists a constant $c(t)>0$ such that that the following
holds. Let $X$ be a finite set of points in general position in the plane;
 and let $X_1,\dots,X_t$ be a partition of $X$. Then there exist subsets
 $X_1' \subseteq X_1, \dots, X_t'\subseteq X_t$ such that $(X_1',\dots,X_t')$  has
 same-type transversals and $|X_i'| \ge c(t) |X_i|$, for all $i=1,\dots, t$. 
\end{theorem}
Both order types and the Same-type lemma
can be  generalized to $\mathbb{R}^d$. For our purposes we only need
 the planar case.  We use the Same-type lemma to show the existence of
 the $Y_i$ and $Z_i$ subsets in condition $a)$ of Lemma~\ref{lem:main}.

\subsection*{The Erd\H{o}s-Simonovits Theorem}

 To show condition $d)$ of Lemma~\ref{lem:main} we need to show that many of the $Y_i-Z_i$ edges do not cross.
 Let $H$ be the bipartite geometric graph with partition $(Y_i,Z_i)$ and whose edge set is equal to $E(Y_i,Z_i)$.
  For every subgraph of $H$ isomorphic
to $K_{2,2}$ we get at least a pair of non-crossing edges. We want to find many
copies of $K_{2,2}$ in $H$.
\begin{theorem}[Erd\H{o}s-Simonovits theorem]\label{thm:es_sim}
Let $t$ be a positive integer and let $G$ be a graph on $n$ vertices and 
with $d \binom{n}{2}$ edges, where $d \ge C n^{-1/t^2}$ for a certain
sufficiently large constant $C$. Then $G$ contains at least
\[c d^{t^{2}} n^{2t}\]
copies of $K_{t,t}$, where $c=c(t) > 0$ is a constant.
\end{theorem}
The Erd\H{o}s-Simonovits theorem was proved in~\cite{es_sim}. Where it is stated
for uniform hypergraphs. We adapted the exposition of~\cite{mat} for the case
of ordinary graphs.
Theorem~\ref{thm:es_sim} implies that if a graph has $c n^2$ edges then
it has at least $c'n^4$ copies of $K_{2,2}$ for some constant $c'$ depending
on $c$. In addition, by the Crossing number theorem, it also has $\Theta(n^4)$
pairs of crossing edges. 
We have that

\MyQuote{ the existence of $c_2$ and condition $b)$ in Lemma~\ref{lem:main},
imply the existence of $c_3$ and condition $d)$ in Lemma~\ref{lem:main}.}

\subsection*{Szemer\'edi's Regularity Lemma}

The tool we need to prove condition $b)$ of Lemma~\ref{lem:main} is a variant of the celebrated Szemerédi's regularity lemma.
 Let $A,B$ be two disjoint subsets of vertices of $G$. The \emph{density of the
 pair} $(A,B)$ is defined as \[d(A,B):=\frac{|E(A,B)|}{|A||B|}.\]
Let $\epsilon >0$. The pair $(A,B)$ is called an
\emph{$\epsilon$-regular pair} if it satisfies the following.
For all $X \subset A$ and $Y \subset B$, such that
\[|X| \ge \epsilon |A| \textrm{ and } |Y| \ge \epsilon|B|,\]
we have that 
\[|d(X,Y)-d(A,B)| \le \epsilon. \]
Let $P:=\{V_0,V_1,\dots,V_t\}$ be a partition of $V$ in which $V_0$
is allowed to be empty. We call $P$ an \emph{$\epsilon$-regular partition} of $G$
if it satisfies the following properties.
\begin{enumerate}
 \item $|V_0| \le \epsilon n;$
 \item $|V_1|=\cdots=|V_t|;$
 \item all but at most $\epsilon t^2$ of the pairs $(V_i,V_j)$ are $\epsilon$-regular.
\end{enumerate}

In 1975, Szemer\'edi~\cite{regularity_lemma} proved the following fundamental
result in extremal graph theory.
\begin{theorem}[Szemer\'edi's regularity lemma]
 For every $\epsilon >0$ and every integer $m \ge 1$ there exists an integer $M$ such
 that the following holds. Every graph on $n \ge m$ vertices admits an $\epsilon$-regular
 partition $\{V_0,V_1,\dots,V_t\}$ with $m \le t \le M$.
\end{theorem}

\subsubsection*{Regularity lemma for multipartite graphs}

Duke, Lefmann and Rödl~\cite{multi_regular} proved a version of the Regularity Lemma for multipartite graphs. We use this result
to show condition $b)$ of Lemma~\ref{lem:main}. For a more recent account of this result see the survey of Rödl and Schacht~\cite{reg_survey}.

Suppose that $G$ is an $r$-partite graph with vertex partition equal to $\{V_1,\dots,V_r\}$, and that every $V_i$ has
cardinality equal to $m$. For every $1 \le i \le r$, let $W_i \subset V_i$. We call the set of tuples $W_1 \times \cdots \times W_r$
a \emph{box}.\footnote{In \cite{multi_regular} they prefer the term \emph{cylinder}.} A box $W_1 \times \cdots \times W_r$
is called \emph{$\epsilon$-regular}, if for every $1 \le i < j \le r$, the pair $(W_i,W_j)$ is $\epsilon$-regular.
Let $\mathcal{P}$ be a partition of $V_1\times \cdots \times V_r$ into boxes. We say that $\mathcal{P}$ is an 
\emph{$\epsilon$-regular}
partition of $V_1\times \cdots \times V_r$ if all but at most $\epsilon m^r$ of the tuples 
$(v_1,\dots,v_r) \in V_1\times \cdots \times  V_r$
lie in non $\epsilon$-regular boxes. The result of~\cite{multi_regular} states that for every fixed $\epsilon>0$ there always exists an $\epsilon$-regular partition of $V_1\times \cdots \times V_r$ into  boxes, in which every box is not too small. The number
of such boxes is a function only of $\epsilon$ and $r$.
\begin{theorem}[Regularity lemma for multipartite graphs]\label{lem:multi_regular}
Let $G$ be an $r$-partite graph with vertex partition equal to $\{V_1,\dots,V_r\}$ and such that every $V_i$ has
cardinality equal to $m$. For every $\epsilon >0$ there exists an $\epsilon$-regular partition $\mathcal{P}$ of 
$V_1\times \cdots \times V_r$ such that:
\begin{enumerate}
 \item $|\mathcal{P}| \le 4^{r^2/\epsilon^5}$; and
 \item for every $W_1\times \cdots \times W_r \in \mathcal{P}$, and every $1 \le i \le r$ we have that
 \[|W_i| \ge \epsilon^{r^2/\epsilon^5}m.\]
\end{enumerate}
\end{theorem}

To apply the regularity lemmas we need to reason about
$\epsilon$-regular partitions of graphs.
For this purpose we define some graphs and extend our definition of density to tuples
and boxes. Let $v=(v_1,\dots,v_r) \in V_1 \times \cdots \times V_r$, and
let $G[v]$ be the subgraph of $G$ induced by the set of vertices $\{v_1,\dots,v_r\}$. We define the
\emph{density}, $d(v)$,  of the tuple $v$ as the density of $G[v]$. Thus,
\[d(v)=\frac{||G[v]||}{\binom{r}{2}}.\]

Let $W=W_1 \times \cdots \times W_r \in \mathcal{P}$ . We define
the \emph{density} of $W$ as
\[d(W):=\frac{\sum_{v \in W}d(v)}{|W|}.\]
For $0 \le \delta \le 1$, let $R(W,\delta)$
be the graph whose vertex set is equal to $\{W_1,\dots,W_r\}$, and in which $W_i$ is adjacent to $W_j$
if the density of the pair $(W_i,W_j)$ is at least $\delta$.

\begin{lemma}\label{lem_R}
Let $G$ be a dense graph on $n$ vertices. Then for every $ 0 <\epsilon \le d(G)/2$ and 
every positive integer $r$, there exist a set $\{W_1,\dots,W_r\}$ of disjoint subsets of vertices
of $G$, such that the following hold. 
\begin{itemize}
 \item[$(1)$] $|W_i| \ge \epsilon^{r^2/\epsilon^5} \frac{n}{r},$  for every $i=1,\dots,r$;
 
 \item[$(2)$] $W:=W_1\times \cdots \times W_r$ is $\epsilon$-regular; and

 \item[$(3)$] $ d \left ( R \left (W,\frac{d(G)}{4-d(G)}\right )  \right )  \ge \frac{d(G)}{4}.$
\end{itemize}
\begin{proof}
Assume that $V(G):=\{1,\dots,n\}$
 If necessary, iteratively remove minimum degree vertices from $G$ so that the number
 of vertices remaining is a multiple of $r$. Note that these operations do not decrease
 the density of $G$. In what follows we assume that $n$ is divisible by $r$. 
 
 Let $V_1,\dots,V_r$ a partition of the vertices of $G$, chosen uniformly at random
 among all the partitions of the vertices of $G$ into $r$ sets of cardinality $n/r$ each.
 Let $G'$ be the $r$-partite graph with partition $V_1,\dots,V_r$, in which $v \in V_i$
 is adjacent to $w \in V_j$ if $vw$ is an edge of $G$. Let $A:=\{a_1,\dots,a_r\}$ be a 
 set of $r$ vertices of $G$. Let $E_A$ be the event that there exists a tuple
 $v=(v_1,\dots,v_r) \in V_1 \times \cdots \times V_r$ such that
 $A=\{v_1,\dots,v_r\}$.

 We can compute $\operatorname{Prob}(E_A)$ by considering the partition of $V(G)$ given by 
  \[\{1,\dots,n/r\},\{n/r+1,\dots,2(n/r)\},\dots \{n-r+1,\dots,n\}.\]
 Let $\sigma$ be a random permutation of the vertices of $G$.
 Note that  \[\{\sigma(1),\dots,\sigma(n/r)\},\{\sigma(n/r+1),\dots,\sigma(2(n/r))\},\dots \{\sigma(n-r+1),\dots,\sigma(n)\}\] produces a random partition of the vertices of $G$, chosen uniformly at random
 among all the partitions of the vertices of $G$ into $r$ sets of cardinality $n/r$ each.
 The number of permutations in which the $\sigma(a_i)$ lie in different sets of the partition
 is equal to $r!(n/r)^r(n-r)!$. Therefore, 
  \[\operatorname{Prob}(E_A)=\frac{r! (n/r)^r(n-r)!}{n!}=\left ( \frac{n}{r} \right )^r  \binom{n}{r}^{-1}.\]
  
Since the endpoints of every edge of $G$ lie in $\binom{n-2}{r-2}$ subsets of $V(G)$ of cardinality $r$, we have
that
\begin{align*}
\sum_{\substack{A \subset V(G) \\ |A|=r }} d(G[A]) & =\sum_{\substack{A \subset V(G) \\ |A|=r }} \binom{r}{2}^{-1} ||G[A]|| \\ 
& =\binom{r}{2}^{-1} \binom{n-2}{r-2} ||G|| \\
& =\binom{r}{2}^{-1} \binom{n-2}{r-2} \binom{n}{2} d(G) \\
&\frac{2\cdot (n-2)! \cdot n \cdot (n-1) }{r\cdot (r-1) \cdot (n-2-(r-2))! \cdot (r-2)! \cdot 2} \cdot d(G) \\
& = \binom{n}{r} d(G).
\end{align*}
Let $X_A$ be the indicator random variable  associated to $E_A$. By linearity of 
expectation we have that 
\begin{align*}
E[d(V_1 \times \cdots \times V_r)] &=\left ( \frac{n}{r} \right)^{-r}\sum_{\substack{A \subset V(G) \\ |X|=r }} E[X_A] d(G[A]) \\
& =\left ( \frac{n}{r} \right)^{-r}\left ( \frac{n}{r} \right )^r \cdot \binom{n}{r}^{-1}\sum_{\substack{A \subset V(G) \\ |X|=r }} d(G[A]) \\
&= d(G).
\end{align*}
Thus, there exist a choice for $V_1,\dots, V_r$ is such that $d(V_1,\times \cdots \times V_r)\ge d(G).$
In what follows, assume that this is the case.

Let $\mathcal{P}$ be the $\epsilon$-partition of $V_1 \times \cdots \times V_r$  given by Lemma~\ref{lem:multi_regular}.
Let $\mathcal{P}'$ be the set of boxes in $\mathcal{P}$ that are $\epsilon$-regular.
Since $\mathcal{P}$ is $\epsilon$-regular, we have that
\begin{align*} 
\sum_{W \in \mathcal{P}'} d(W) |W|& =\sum_{W \in \mathcal{P}} d(W) |W|-\sum_{W \in \mathcal{P}\setminus \mathcal{P'}} d(W) |W| \\
& \ge d(G)\left( \frac{n}{r}\right)^r-\epsilon \left (\frac{n}{r} \right)^{r} \\
&= (d(G)-\epsilon)\left ( \frac{n}{r} \right )^r \\
&\ge  \frac{d(G)}{2} \left ( \frac{n}{r} \right )^r.
\end{align*}
Suppose that for all $W \in \mathcal{P}'$
we have that $d(W) < d(G)/2$. Thus,
\[\sum_{W \in \mathcal{P}'} d(W) |W| < \frac{d(G)}{2}\sum_{W \in \mathcal{P}'} |W| \le \frac{d(G)}{2} \left ( \frac{n}{r} \right)^r; \]
this is a contradiction. Therefore, there exist 
$W=W_1\times \cdots \times  W_r \in \mathcal{P}'$ such that $d(W) \ge d(G)/2$.

Let $xy \in E(W_i,W_j)$ for some $1 \le i <j \le r$. Note that there are exactly
\[\prod_{l \neq i,j}|W_l|\]
tuples $u \in W$ such that $xy \in E(G[u])$. This implies that

\begin{align*}
 \sum_{v \in W} d(v)& =\sum_{v \in W} \frac{||G[u]||}{\binom{r}{2}} \\
 &= \binom{r}{2}^{-1}\sum_{1 \le i <j \le r} \left ( |E(W_i,W_j)|\prod_{l \neq i,j}|W_l| \right ) \\
 &= \binom{r}{2}^{-1}\sum_{1 \le i <j \le r} \left ( d(W_i,W_j) |W_i||W_j|\prod_{l \neq i,j}|W_l| \right ) \\
 &= \binom{r}{2}^{-1}\sum_{1 \le i <j \le r}  d(W_i,W_j) |W|.
\end{align*}
Therefore,
\[\sum_{1 \le i <j \le r}d(W_i,W_j)= \binom{r}{2} d(W) \ge \frac{d(G)}{2} \binom{r}{2}.  \]

Let $E'=E(R(W,d(G)/(4-d(G)))$.
We have that
\begin{align*}
 \frac{d(G)}{2} \binom{r}{2} & \le \sum_{1 \le i <j \le r}d(W_i,W_j) \\
 &=\sum_{W_iW_j \in E'}d(W_i,W_j)+\sum_{W_iW_j \notin E'}d(W_i,W_j)\\
 & \le \frac{d(G)}{4-d(G)}\left (\binom{r}{2}-|E'| \right )+|E'| \\
 &= \left( 1- \frac{d(G)}{4-d(G)} \right )|E'|+\frac{d(G)}{4-d(G)}\binom{r}{2}.
\end{align*}

Therefore,
\begin{align*}
 |E'| & \ge \left ( \left (  \frac{d(G)}{2}-\frac{d(G)}{4-d(G)}\right) \bigg /\left( 1- \frac{d(G)}{4-d(G)} \right ) \right ) \binom{r}{2}\\
 & = \left ( \left (  \frac{d(G)(4-d(G))-2d(G)}{2(4-d(G))}\right) \bigg /\left(\frac{4-d(G)-d(G)}{4-d(G)}\right ) \right ) \binom{r}{2}\\
 & = \left ( \left (  \frac{d(G)(2-d(G))}{2(4-d(G))}\right) \bigg /\left(\frac{2(2-d(G))}{4-d(G)}\right ) \right ) \binom{r}{2}\\
 &=\frac{d(G)}{4} \binom{r}{2}.
\end{align*}
The result follows.

\end{proof}

\end{lemma}


\subsection*{Pairwise crossing edges in geometric graphs}

Aronov, Erd\H{o}s, Goddard, Kleitman, Klugerman, Pach and Schulman~\cite{crossing_families}
showed that every complete geometric graph on $n$ vertices contains $\sqrt{n}/12$ pairwise crossing edges. 
Pach, Rubin and Tardos~\cite{many_crossing}  improved and generalized this bound.
They showed that every dense graph on $n$ vertices contains 
$n^{1-o(1)}$ pairwise crossing edges. Recently, the constant in the O-notation was improved
for complete geometric graphs by Suk and Zeng~\cite{positive}.

It is an open problem (Chap 9, Problem 1~\cite{research_problems}) to show that
for every positive integer $k>3$ there exists a constant $c_k >0$ such that every geometric graph on $n$ vertices
and more than $c_kn$ edges contains $k$ pairwise crossing edges. 
In this direction Valtr~\cite{pavel}, showed the following result.
\begin{theorem}\label{thm:pavel}
 Let $k$ be positive integer. A geometric graph on $n$ vertices without $k$ pairwise crossing
 edges contains at most $O(n \log n)$ edges.
\end{theorem}

We are ready to prove Theorem~\ref{thm:main}.

\section{proof of Theorem~\ref{thm:main}}

Let $G=(V,E)$ be a dense geometric graph on $n$ vertices, with $n$ sufficiently large, and density $d>0$, and let $k \ge 2$ be an integer. 
To prove Theorem~\ref{thm:main}, we show that

\NMyQuote{there exist positive constants $c_1,c_2$ and $c_3 < c_2/2$, such that conditions $a)$,$b)$,$c)$ and $d)$ of Lemma~\ref{lem:main} hold.}

By Theorem~\ref{thm:pavel}, there exists a positive integer $r$ (depending only on $d$ and $k$) such that every geometric graph on $r$ or more vertices, of
density at least $d/4$ contains $k$ pairwise crossing edges. Let $c(r)$ be as in the Same-type lemma. Let 
\[0 \le \epsilon < \min \left \{c(r),\frac{d}{4-d}-\frac{d}{4} \right \}.\] 
Simple arithmetic shows that since $\epsilon \le \frac{d}{4-d}-\frac{d}{4}$, we have
that $\epsilon <d/2$. 
Let $W_1,\dots,W_r$ be the disjoint subsets of vertices of $G$ given 
by Lemma~\ref{lem:multi_regular}. By the Same-type lemma there exist  $W_1'\subset W_1, \dots, W_r'\subset W_r$ such
that $(W_1',\dots,W_r')$ has same type transversals and $|W_i'|\ge c(r) |W_i|$ for all $i=1,\dots,r$. 
Let \[c_1:= \frac{c(r)\epsilon^{r^2/\epsilon^5}}{r}.\]
Thus,
\[|W_i'|\ge c(r) |W_i| \ge \frac{c(r)\epsilon^{r^2/\epsilon^5}}{r}  n = c_1 n.\]

Let $(u_1,\dots,u_r)$ a transversal of $(W_1',\dots,W_r')$. Let $G'$ be the geometric graph whose vertex set
is equal to $\{u_1,\dots,u_r\}$; in which $u_i$ is adjacent to $u_j$, if 
\[d(W_i,W_j) \ge \frac{d}{4-d}.\]
By $(3)$ of Lemma~\ref{lem:multi_regular} and our choice of $r$, we have that $G'$ contains
$k$ pairwise crossing edges, $e_1,\dots,e_k$.
For every $1 \le i \le k$, let $Y_i,Z_i \in \{W_1',\dots, W_r'\}$ such that
$e_i$ has an endpoint in $Y_i$ and an endpoint in $Z_i$. This proves condition $a)$. 

Let \[c_2 :=\frac{d}{4} c_1^2.\] Let $1 \le i < j \le r$.
Since $W_1\times \cdots \times W_r$ is $\epsilon$-regular, $\epsilon < c(t)$, and $e_i$ is an edge
of $G'$, we have that
\[d(X_i,Y_i) \ge \frac{d}{4-d}-\epsilon \ge \frac{d}{4}\]
Thus,
\[E(X_i,Y_i) \ge \frac{d}{4} |X_i||Y_i| \ge  \frac{d}{4} c_1^2 n^2=c_2 n^2.\]
This proves condition $b)$. Since the $e_i$ are pairwise crossing and  $(W_1',\dots,W_r')$ has
same type transversals we have condition $c)$.
Finally, as noted above ($\ast$), condition $d)$ and the existence of  $c_3$ follows from condition $b)$ and the existence of $c_2$.
This completes the proof of Theorem~\ref{thm:main}.

\section*{Acknowledgments}

I thank Irene Parada and Birgit Vogtenhuber for various helpful discussions. I also thank the anonymous
reviewer who found a crucial flaw in a previous version of this paper.


\small \bibliographystyle{alpha} \bibliography{refs_mono_crossings}

\end{document}